\documentclass[reqno, a4paper]{amsart}

\usepackage{amsmath, amssymb,amsthm, mathrsfs}
\usepackage{txfonts}
\usepackage{enumerate}
\usepackage{color}
\usepackage{bbm}
\usepackage[hyperfootnotes=false]{hyperref}
\newtheorem{theorem}{Theorem}

\newtheorem{definition}[theorem]{Definition}
\newtheorem{lemma}[theorem]{Lemma}
\newtheorem{proposition}[theorem]{Proposition}
\newtheorem{convention}[theorem]{Important convention}
\newtheorem{remark}[theorem]{Remark}

\newcommand{\OB}{Ob{\l}{\'o}j}

\newcommand{\eps}{\varepsilon}

\newcommand{\F}{\mathcal{F}}

\newcommand{\W}{\mathcal{W}}

\newcommand{\E}{\mathbb{E}}
\newcommand{\N}{\mathbb{N}}
\newcommand{\Q}{\mathbb{Q}}
\newcommand{\R}{\mathbb{R}}
\newcommand{\Z}{\mathbb{Z}}

\newcommand{\G}{\mathcal G}

\newcommand{\cG}{\mathcal{G}}

\newcommand{\cE}{\mathcal{E}}
\newcommand{\dd}{\,\mathrm{d}}

\newcommand{\1}{\mathbf 1}

\newcommand{\Leb}{\mathcal{L}}
\newcommand{\dds}{\mathrm{d}}
\newcommand{\CS}{\S}

\renewcommand{\W}{{\mathbb W}}
\renewcommand{\eps}{\varepsilon}
\renewcommand{\epsilon}{\varepsilon}
\renewcommand{\phi}{\varphi}

\renewcommand{\P}{\mathbb{P}}

\renewcommand{\S}{\Upsilon}
\renewcommand{\d}{\mathrm{d}}

\DeclareMathOperator{\supp}{supp}

\numberwithin{equation}{section}
\numberwithin{theorem}{section}

\newcommand{\RT}{\mathsf{RT}}

\newcommand{\RST}{\mathsf{RST}}

\newcommand{\TRST}{\mathsf{JOIN}}

\newcommand{\ol}[1]{\bar{#1}}

\renewcommand{\subset}{\subseteq}

\renewcommand{\llcorner}{\upharpoonright}

\DeclareMathOperator{\ntt}{\mathfrak{t}}
\newcommand{\tntt}{\bar\ntt}

\title{Pathwise super-replication via Vovk's outer measure}

\author[Beiglb\"ock]{Mathias Beiglb\"ock}
\address{Mathias Beiglb\"ock, Faculty of Mathematics, University of Vienna, Austria}
\email{mathias.beiglboeck@univie.ac.at}

\author[Cox]{Alexander M. G. Cox}
\address{Alexander M. G. Cox, Department of Mathematical Sciences, University of Bath, United Kingdom}
\email{A.M.G.Cox@bath.ac.uk}

\author[Huesmann]{Martin Huesmann}
\address{Martin Huesmann, Institut f\"ur Mathematik, Rheinische Friedrich-Wilhelms-Universit\"at Bonn, Germany}
\email{huesmann@iam.uni-bonn.de}

\author[Perkowski]{Nicolas Perkowski}
\address{Nicolas Perkowski, Institut f\"ur Mathematik, Humboldt-Universit\"at zu Berlin, Germany}
\email{perkowsk@math.hu-berlin.de}

\author[Pr\"omel]{David J. Pr\"omel}
\address{David J. Pr\"omel, Department of Mathematics, Eidgen\"ossische Technische Hochschule Z\"urich, Switzerland}
\email{david.proemel@math.ethz.ch}

\date{\today}

\begin{document}

\begin{abstract} 
  Since Hobson's seminal paper~\cite{Ho98a} the connection between model-independent pricing and the Skorokhod embedding problem has been a driving force in robust finance. We establish a general pricing-hedging duality for financial derivatives which are susceptible to the Skorokhod approach.
  
  Using Vovk's approach to mathematical finance we derive a model-independent super-replication theorem in continuous time, given information on finitely many marginals.
  Our result covers a broad range of exotic derivatives, including lookback options, discretely monitored Asian options, and options on realized variance.\medskip
  
  %
  % 60G44 ~ Stochastic processes (Martingales with continuous parameters)
  % 91G20 ~ Mathematical finance (Derivative securities)
  % 91B24 ~ Mathematical economics (Price theory and market structure) 
\end{abstract}

\maketitle

\let\thefootnote\relax\footnote{The authors are grateful for the excellent hospitality of the Hausdorff Research Institute for Mathematics (HIM), where the work was carried out. M.B.\ gratefully acknowledges support through FWF-projects P26736 and Y782-N25. M.H.\ gratefully acknowledges the financial support of the CRC 1060. D.J.P.\ gratefully acknowledges the financial support of the DFG Research Training Group 1845  and the Swiss National Foundation under Grant No.~$200021\_163014$.}

\section{Introduction}

Initiated by Hobson~\cite{Ho98a}, the theory of model-independent pricing has received substantial attention from the mathematical finance community, we refer to the survey~\cite{Ho11}. Starting with \cite{BeHePe12, GaHeTo13}, the Skorokhod embedding approach has been complemented through optimal transport techniques. In particular, first versions of a robust super-replication theorem have been established: in discrete time we mention \cite{AcBePeSc13} and the important contribution of Bouchard and Nutz~\cite{BoNu13}; for related work in a quasi-sure framework in continuous time we refer to the work of Neufeld and Nutz~\cite{NeNu13} and Possama\"i, Royer, and Touzi~\cite{PoRoTo13}. Our results are more closely related to the continuous time super-replication theorem of Dolinsky and Soner~\cite{DoSo12}, which we recall here: given a centered probability measure~$\mu$ on $\R$, they study the primal maximization problem
\begin{equation*}
  P:= \sup\big\{ \E_\P [G(S)]\big\}
\end{equation*}
where $S$ denotes the canonical process on $C[0,1]$, the supremum is taken over all martingale measures $\P$ on $C[0,1]$ with $S_1(\P)=\mu$ and $G$ denotes a functional on the path space satisfying appropriate continuity assumptions. The main result of \cite{DoSo12} is a super-replication theorem that appeals to this setup: they show that for each $p>P$ there exists a hedging strategy~$H$ and a ``European payoff function''~$\psi$ with $\int \psi \dd\mu=0$ such that 
\begin{equation*}
  p + (H\cdot S)_1 + \psi(S_1) \geq G(S).
\end{equation*}
This is in principle quite satisfying, however, a drawback is that the option $G$ needs to satisfy rather strong continuity assumptions, which in particular excludes all exotic option payoffs involving volatility. Given the practical importance of volatility derivatives it is desirable to give a version of the Dolinsky-Soner theorem that appeals also to this case. More recently, Dolinsky and Soner~\cite{DoSo14} have extended the original results of \cite{DoSo12} to include c\`adl\`ag price processes, multiple maturities and price processes in higher dimensions; Hou and Ob\l\'oj~\cite{HoOb15} have also recently extended these results to incorporate investor beliefs via a `prediction set' of possible outcomes.

Subsequently, we shall  establish a super-replication theorem that applies to $G$ which is invariant under time-changes in an appropriate sense. Opposed to the result of \cite{DoSo12} this excludes the case of continuously monitored  Asian options but covers other practically relevant derivatives such as options on volatility or realized variance, lookback options and discretely monitored Asian options. Notably, it constitutes a general duality result appealing to the rich literature on the connection of model-independent finance and Skorokhod embedding. In a series of  impressive achievements, Brown, Cox, Davis, Hobson, Klimmek, Neuberger, \OB, Pedersen, Raval, Rogers,  Wang, and others \cite{Ro93, Ho98a,BrHoRo01a, HoPe02, CoHoOb08, DaObRa10, CoOb11a,CoOb11b, CoWa13,  HoNe12,HoKl12} were able to determine the values of related primal and dual problems for a number of exotic derivatives/market data, proving that they are equal. Here we establish the duality relation for generic derivatives, in 
particular recovering duality for the specific cases mentioned above.

 To achieve this we apply a pathwise approach to model-independent finance which was introduced by Vovk~\cite{Vo11a,Vo12,Vo15}. In particular we rely on Vovk's pathwise Dambis Dubins-Schwarz theorem, which we combine with the duality theory for the Skorokhod embedding problem recently developed in~\cite{BeCoHu14}.

 \medskip 

After the completion of this work, we learned that Guo, Tan, and Touzi~\cite{GuTaTo15} derived a duality result similar in spirit to Theorem~\ref{thm:duality}. Their approach relies on different methods, and includes an interesting application to the optimal Skorokhod embedding problem.\medskip

\noindent{\bf Organization of the paper:} In Section~\ref{appetizer} we state our main result. In Section~\ref{sec:setting} Vovk's approach to mathematical finance is introduced and preliminary results are given. Section~\ref{sec:result} is devoted to the statement and proof of our main result in its simplest form: a super-replication theorem for time-invariant payoffs for one period. In Section~\ref{sec:duality} we present an extension to finitely many marginals with ``zero up to full information''; in particular we will then obtain our most general super-replication result, Theorem~\ref{thm:generalduality}.

\section{Formulation of the super-replication theorem}\label{appetizer}

For $n\in\mathbb{N}$ let $C[0,n]$ be the space of continuous functions $\omega \colon [0,n] \to \mathbb{R}$ with $\omega(0)=0$ and consider $G\colon C[0,n]\to \R$ of the form
\begin{equation}\label{eq:G-time-inv}
  G(\omega)=\gamma( \ntt (\omega)_{\llcorner [0, \langle \omega\rangle_{n}] }, \langle \omega\rangle_{1}, \ldots, \langle \omega\rangle_{n}),
\end{equation}
where $ \langle \omega\rangle_{\cdot}$ stands for the quadratic variation process of the path $\omega$ and $\ntt (\omega)$ stands for a version of the path $\omega$ which is rescaled in time so that for each $t$ its quadratic variation up to time $t$ equals precisely $t$. Let $S$ be the canonical process on $C[0,n]$. Under appropriate regularity conditions on $\gamma$ (see Theorems~\ref{HedgingDual} and \ref{thm:generalduality} below) we obtain the following robust super-hedging result: 

\begin{theorem}%\label{thm:generalduality:intro}
  Suppose that $n\in\mathbb{N}$, $I\subseteq \{1,\ldots, n\}$, $n\in I$ and that $\mu_i$ is a centered probability measure on $\mathbb{R}$ for each $i\in I$. Setting  
  \begin{align*}
    P_n:=\sup\big\{ \E_\P[G]: \text{$\P$ is a martingale measure on $C[0,n]$, $S_0=0$,  $S_{i}\sim \mu_i$  for all  $i\in I$}\big\}
  \end{align*}
  and 
  \begin{align*}
    D_n:=\inf\left\{a\,:\, \begin{array}{l} \text{ there exist } H \text{ and } (\psi_j)_{j\in I}\text{ s.t. } \int \psi_j \dd \mu_j=0,\\
                                            a +\sum_{j\in I} \psi_j(S_j) + (H\cdot S)_n \geq G((S_t)_{t\leq n})
                           \end{array} \right\},
  \end{align*}
  one has $P_n=D_n$. Here $(H\cdot S)_n$ denotes the ``pathwise stochastic integral'' of $H$ with respect to $S$.
\end{theorem}

Of course the present statement of our main result is imprecise in that neither the pathwise stochastic integral appearing in the formulation of $D_n$, nor the pathwise quadratic variation in the definition of $G$, nor the function $\ntt$ are properly introduced. We will address this in the following sections.

Examples of derivatives in the time-invariant form~\eqref{eq:G-time-inv} include the following:
\begin{itemize}
   \item[--] $G_1(\omega) = F_1(\omega(1), \dots, \omega(n), \langle \omega \rangle_1,  \ldots, \langle \omega\rangle_{n})$;
   \item[--] $G_2(\omega) =  F_2(\max_{t \in [0,n]} \omega(t))$;
   \item[--] $G_3(\omega) = F_3(\int_0^n \varphi(\omega(s), \langle \omega \rangle_s) \dd \langle \omega \rangle_s)$;
   \item[--] $G_4(\omega) = F_4(G_1(\omega), G_2(\omega), G_3(\omega))$.
\end{itemize}
Examples that are not covered by our results are continuously monitored Asian options, $G(\omega) = F(\int_0^n \omega(s) \dd s)$. In that case we would have to discretize time and consider the discretely monitored version $\tilde{G}(\omega) = F(\sum_{k=0}^{n-1} \omega(k))$.

\section{Super-hedging and outer measure}\label{sec:setting}

Very recently, Vovk~\cite{Vo11a,Vo12,Vo15}, see also~\cite{TaKuTa09}, developed a new model free approach to mathematical finance based on hedging. Without presuming any probabilistic structure, Vovk considers the space of real-valued continuous functions as possible price paths and introduces an outer measure on this space, which is based on a minimal super-hedging price.

Vovk defines his outer measure on all continuous paths, and then shows that ``typical price paths'' admit a quadratic variation. To simplify many of the statements and proofs below, we restrict ourselves from the beginning to paths admitting quadratic variation. We discuss in Remark~\ref{rem:restriction to qv paths} below why this is no problem.

To be precise, define for a continuous path $\omega \colon \R_+ \to \R$ and $n\in \mathbb{N}$ the stopping times
\begin{align*}
  \sigma_0^n:=0 \quad \text{and}\quad
  \sigma_k^n:= \inf \{t\geq \sigma_{k-1}^n\,:\,\omega(t) \in 2^{-n}\Z \text{ and } \omega(t)\neq \omega (\sigma_{k-1}^n)\},
\end{align*}
for $k\in \mathbb{N}$. For $n\in \mathbb{N}$ the discrete quadratic variation of $\omega$ is given by
\begin{equation*}
  V^n_t(\omega) :=\sum_{k=0}^\infty \big(\omega({\sigma^n_{k+1}\wedge t})-\omega({\sigma_{k}^n\wedge t})\big)^2,\quad t\in \mathbb{R}_+.
\end{equation*}
We write $\Omega^{\mathrm{qv}}$ for the space of continuous functions $\omega \colon \R_+ \to \mathbb{R}$ with $\omega(0)=0$ such that $V^n(\omega)$ converges locally uniformly in time to a continuous limit $\langle \omega \rangle$ which has the same intervals of constancy as $\omega$; moreover, we assume that for every $\omega \in \Omega^{\mathrm{qv}}$ either $\lim_{t\rightarrow \infty} \omega(t)$ exists or $\langle \omega \rangle$ is unbounded on $\R_+$.

The coordinate process on $\Omega^{\mathrm{qv}}$ is denoted by $B_t(\omega):=\omega(t)$ and we introduce the natural filtration $(\mathcal{F}_t^{\mathrm{qv}})_{t\geq 0} := (\sigma(B_s: s \le t))_{t\geq 0}$ and set $\mathcal{F}^{\mathrm{qv}}:= \bigvee_{t\geq 0} \mathcal{F}^{\mathrm{qv}}_t$. Stopping times $\tau$ and the associated $\sigma$-algebras $\F^{\mathrm{qv}}_\tau$ are defined as usual. Occasionally we will also write $\langle B \rangle(\omega) = \langle \omega \rangle$.

A process $H \colon \Omega^{\mathrm{qv}} \times \mathbb{R}_+ \rightarrow \R$ is called a \emph{simple strategy} if it is of the form 
\begin{equation*}
  H_t(\omega) = \sum_{n= 0}^{\infty} F_n(\omega) \1_{(\tau_n(\omega),\tau_{n+1}(\omega)]}(t), \quad (\omega,t)\in \Omega^{\mathrm{qv}}\times \mathbb{R}_+,
\end{equation*}
where $0 = \tau_0(\omega) < \tau_1(\omega) < \dots$ are stopping times such that for every $\omega \in \Omega^{\mathrm{qv}}$ one has $\lim_{n\to \infty} \tau_n(\omega) = \infty$, and $F_n\colon \Omega^{\mathrm{qv}} \rightarrow \R$ are $\F^{\mathrm{qv}}_{\tau_n}$-measurable bounded functions for $n\in\mathbb{N}$. For such a simple strategy $H$ the corresponding capital process
\begin{equation*}
  (H \cdot B)_t(\omega) = \sum_{n=0}^\infty F_n(\omega) (B_{\tau_{n+1}(\omega) \wedge t}(\omega) - B_{\tau_n(\omega) \wedge t}(\omega))
\end{equation*}
is well-defined for every $\omega \in \Omega^{\mathrm{qv}}$ and every $t \in \mathbb{R}_+$. A simple strategy $H$ is called $\lambda$-\textit{admissible} for $\lambda > 0$ if $(H\cdot B)_t(\omega) \ge - \lambda$ for all $t \in \R_+$ and all $\omega \in \Omega^{\mathrm{qv}}$. We write $\mathcal{H}_\lambda$ for the set of $\lambda$-admissible simple strategies.

To recall Vovk's outer measure as introduced in \cite{Vo12}, let us define the set of processes
\begin{align*}
  \mathcal{V}_\lambda := \left\{ G:= \big(H^k\big)_{k\in \N}\,:\, H^k \in \mathcal{H}_{\lambda_k}, \lambda_k > 0, \sum_{k=0}^\infty \lambda_k = \lambda\right\}
\end{align*}
for an initial capital $\lambda \in (0, \infty)$. Note that for every $G = \big(H^k\big)_{k\in \N} \in \mathcal{V}_{\lambda}$, all $\omega \in \Omega^{\mathrm{qv}}$, and all $t \in \mathbb{R}_+$, the corresponding capital process
\begin{align*}
  (G\cdot B)_t(\omega) := \sum_{k = 0}^{\infty} (H^k\cdot B)_t(\omega) = \sum_{k = 0}^{\infty} \big(\lambda_k + (H^k\cdot B)_t(\omega)\big) - \lambda
\end{align*}
is well-defined and takes values in $[-\lambda, \infty]$. 

Then, Vovk's outer measure on $\Omega^{\mathrm{qv}}$ is given by 
\begin{align*}
  \overline{Q}(A) := \inf\Big\{\lambda > 0\,:\, \exists \, G \in \mathcal{V}_\lambda \text{ s.t. } 
                                        \lambda + \liminf_{t \to \infty}(G\cdot B)_t(\omega) \ge \1_A(\omega) \, \forall \omega \in \Omega^{\mathrm{qv}}\Big\}, \quad A \subseteq \Omega^{\mathrm{qv}}.
\end{align*}
A slight modification of the outer measure~$\overline{Q}$ was introduced in \cite{PePr13,PePr15}, which seems more in the spirit of the classical definition of super-hedging prices in semimartingale models. In this context one works with general admissible strategies and the It\^o integral against a general strategy is given as limit of integrals against simple strategies. So in that sense the next definition seems to be more analogous to the classical one.

\begin{definition}
  The \emph{outer measure} $\overline{P}$ of $A \subseteq \Omega^{\mathrm{qv}}$ is defined as the minimal super-hedging price of $\1_A$, that is 
  \begin{align*}
    \overline{P}(A):= \inf\Big\{\lambda > 0: \exists \, (H^n) \subseteq \mathcal{H}_\lambda \text{ s.t.} \liminf_{t\rightarrow\infty} \liminf_{n\rightarrow\infty}\big(\lambda + (H^n\cdot B)_t (\omega)\big) \ge \1_A(\omega)\, \forall \omega \in \Omega^{\mathrm{qv}}\Big\}.
  \end{align*}
  
  A set $A \subseteq \Omega^{\mathrm{qv}}$ is said to be a \emph{null set} if it has  $\overline{P}$ outer measure zero. A property (P) holds for \emph{typical price paths} if the set $A$ where (P) is violated is a null set.
\end{definition}

Of course, for both definitions of outer measures it would be convenient to just minimize over simple strategies rather than over the limit (inferior) along sequences of simple strategies. However, this would destroy the very much appreciated countable subadditivity of both outer measures.
 
\begin{remark}\label{rmk:outer}
  It is conjectured that the outer measure~$\overline{P}$ coincides with $\overline{Q}$. However, up to now it is only known that $\overline{P}(A) \leq \overline{Q}(A)$ for a general set $A \subseteq \Omega^{\mathrm{qv}}$, see \cite[Section~2.4]{PePr13}, and that they coincide for time-superinvariant sets, see Definition~\ref{def:superinvariant} and Theorem~\ref{thm:superinvariant} below. Therefore, the outer measures $\overline{P}$ and $\overline{Q}$ are basically the same in the present paper since we focus on time-invariant financial derivatives. 
\end{remark}

Perhaps the most interesting feature of $\overline{P}$ is that is comes with the following arbitrage interpretation for null sets.

\begin{lemma}[{\cite[Lemma~2.4]{PePr13}}]\label{lem:na1}
  A set $A \subseteq \Omega^{\mathrm{qv}}$ is a null set if and only if there exists a sequence of $1$-admissible simple strategies $(H^n)_{{n\in \N}} \subseteq \mathcal{H}_1$, such that
  \begin{align*}
    1 +\liminf_{t \rightarrow \infty}\liminf_{n \rightarrow \infty} (H^n \cdot B)_t (\omega) \ge \infty \cdot \1_A(\omega),
  \end{align*}
  where we use the convention $\infty \cdot 0 := 0$ and  $\infty \cdot 1 := \infty$.
\end{lemma}

A null set is essentially a model free arbitrage opportunity of the first kind. Recall that $B$ satisfies \emph{(NA1)} (no arbitrage opportunities of the first kind) under a probability measure~$\P$ on $(\Omega^{\mathrm{qv}}, \F^{\mathrm{qv}})$ if the set $\mathcal{W}^{\infty}_1 :=\big \{ 1 + \int_0^\infty H_s \dd B_s\,:\, H \in \mathcal{H}_1\big \}$ is bounded in probability, that is if $\lim_{n \to \infty} \sup_{X\in \mathcal{W}^{\infty}_1} \mathbb{P}( X \geq n)=0$. The notion (NA1) has gained a lot of interest in recent years since it is the minimal condition which has to be satisfied by any reasonable asset price model; see for example~\cite{An05, KaKa07, Ru13, ImPe15}. 

The next proposition collects further properties of $\overline{P}$. 

\begin{proposition}[{\cite[Proposition~3.3]{PePr15}}]\label{prop:properties}
  \begin{enumerate} 
    \item $\overline{P}$ is an outer measure with $\overline{P}(\Omega^{\mathrm{qv}})=1$, i.e. $\overline{P}$ is nondecreasing, countably subadditive, and $\overline{P}(\emptyset) = 0$.
    \item Let $\mathbb{P}$ be a probability measure on $(\Omega^{\mathrm{qv}}, \F^{\mathrm{qv}})$ such that the coordinate process~$B$ is a $\mathbb{P}$-local martingale, and let $A \in \F^{\mathrm{qv}}$. Then $\mathbb{P}(A) \le \overline{P}(A)$.
    \item Let $A \in \F^{\mathrm{qv}}$ be a null set, and let $\mathbb{P}$ be a probability measure on $(\Omega^{\mathrm{qv}}, \F^{\mathrm{qv}})$ such that the coordinate process $B$ satisfies (NA1) under $\P$. Then $\mathbb{P}(A) = 0$.
  \end{enumerate}
\end{proposition}

Especially, the last statement is of interest in robust mathematical finance because it says that every property which is satisfied by typical price paths holds quasi-surely for all probability measures fulfilling~(NA1).

An essential ingredient to obtain our super-replication theorem for time-invariant derivatives is a very remarkable pathwise Dambis Dubins-Schwarz theorem as presented in~\cite{Vo12}. In order to give its precise statement here, we recall the definition of time-superinvariant sets, cf.~\cite[Section~3]{Vo12}.

\begin{definition}\label{def:superinvariant}
  A continuous non-decreasing function $f \colon \R_+ \to \R_+$ satisfying $f(0)=0$ is said to be a \emph{time-change}. The set of all time-changes will be denoted by $\cG_0$, the group of all time-changes that are strictly increasing and unbounded will be denoted by $\cG$. Given $f\in\cG_0$ we define $T_f(\omega):=\omega\circ f.$ A subset $A\subset \Omega^{\mathrm{qv}}$ is called \emph{time-superinvariant} if for all $f\in\cG_0$ it holds that 
  \begin{equation}\label{eq:inv}
    T_f^{-1}(A)\subset A. 
  \end{equation}
  A subset $A\subset \Omega^{\mathrm{qv}}$ is called \emph{time-invariant} if \eqref{eq:inv} holds true for all $f\in\cG$.
\end{definition}

For an intuitive explanation of time-superinvariance we refer to \cite[Remark~3.3]{Vo12}. We denote by $\W$ the Wiener measure on $(\Omega^{\mathrm{qv}}, \F^{\mathrm{qv}})$ and recall Vovk's pathwise Dambis Dubins-Schwarz theorem.
 
\begin{theorem}[{\cite[Theorem~3.1]{Vo12}}]\label{thm:superinvariant}
  Each time-superinvariant set $A\subset \Omega^{\mathrm{qv}}$ satisfies $\overline{P}(A)=\overline{Q}(A)=\W(A)$.
\end{theorem}

\begin{proof}
  For every $A\subset\Omega^{\mathrm{qv}}$ Proposition~\ref{prop:properties} and Remark~\ref{rmk:outer} imply $\W(A)\leq \overline{P}(A) \leq \overline{Q}(A)$.
  If $A$ is additionally time-superinvariant, \cite[Theorem~3.1]{Vo12} says $\overline{Q}(A)=\W(A)$, which immediately gives the desired result. 
\end{proof}

Let us now introduce the \emph{normalizing time transformation} $\ntt$ in the sense of \cite{Vo12}. We follow \cite{Vo12} in defining the sequence of stopping times
\begin{equation}\label{eq:time-chnage}
  \tau_t(\omega):=\inf\left\{s\geq 0: \langle \omega \rangle_s > t\right\}
\end{equation}
for $t\in \mathbb{R}_+$ and $\tau_\infty := \sup_n \tau_n$. The \emph{normalizing time transformation} $\ntt\colon \Omega^{\mathrm{qv}} \to \Omega^{\mathrm{qv}}$ is given by
\begin{equation}\label{eq:qv-inverse}
  \ntt(\omega)_t:= \omega (\tau_t), \quad t\in \mathbb{R}_+, 
\end{equation}
where we set $\omega(\infty) := \lim_{t\to \infty} \omega(t)$ for all $\omega \in \Omega^{\mathrm{qv}}$ with $\sup_{t\ge 0} \langle \omega \rangle_t < \infty$. Note that $\ntt(\omega)_\cdot$ stays constant from time $\langle \omega \rangle_\infty$ on (which is of course only relevant if that time is finite). Below we shall also use $\ntt\colon C_{\mathrm{qv}}[0,1]\to \Omega^{\mathrm{qv}}$ which is defined analogously and where $C_{\mathrm{qv}}[0,1]$ denotes the space of paths that are obtained by restricting functions in $\Omega^{\mathrm{qv}}$ to $[0,1]$. On the product space $\Omega^{\mathrm{qv}} \times \R_+$ we further introduce
\begin{equation*}
  \tntt(\omega, t):= (\ntt(\omega), \langle \omega \rangle_t).
\end{equation*}

We are now ready to state the main result of \cite{Vo12}:

\begin{theorem}[{\cite[Theorem~6.4]{Vo12}}]\label{thm:vovk-dds}
  For any non-negative Borel measurable function \\$F\colon\Omega^{\mathrm{qv}}\to \R$, one has 
  \begin{equation*}
    \overline{\E}[F\circ \ntt,\, \langle B \rangle_\infty=\infty]=\int_{\Omega_\mathrm{qv}} F \dd \W,
  \end{equation*}
  where $\overline{\E}$ is the obvious extension of $\overline{P}$ from sets to nonnegative functions %, $F(\ntt(\omega)):=0$ for all $\omega\notin \ntt^{-1}(\Omega)$,
  and $\langle B \rangle_\infty := \sup_{t\ge 0} \langle B \rangle_t$.
\end{theorem}

\begin{remark}\label{rem:restriction to qv paths}
  It might seem like a strong restriction that we only deal with paths in $\Omega^{\mathrm{qv}}$ rather than considering all continuous functions, however Vovk's result holds on all of $C(\R_+)$, the continuous paths on $\R_+$ started in 0, and is only slightly more complicated to state it in that case. In particular, Vovk shows that $C(\R_+) \setminus \Omega^{\mathrm{qv}}$ is atypical in the sense that for every $\varepsilon>0$ there exists a sequence of $\varepsilon$-admissible simple strategies $(H^n)$ on $C(\R_+)$ (which are defined in the same way as above, replacing every occurrence of $\Omega^{\mathrm{qv}}$ by $C(\R_+)$) such that for every $\omega \in C(\R_+) \setminus \Omega^{\mathrm{qv}}$ we have $\liminf_{t\rightarrow \infty} \liminf_{n \to \infty} (H^n \cdot B)_t (\omega) = \infty$. In particular, all our results continue to hold on $C(\R_+)$ because on the set $C(\R_+) \setminus \Omega^{\mathrm{qv}}$ we can superhedge any functional starting from an arbitrarily small $\varepsilon>0$. To simplify the presentation we restricted our attention to $\Omega^{\mathrm{qv}}$ from the beginning.
\end{remark}

\begin{remark}
   Vovk defines the normalizing time transformation slightly differently, replacing $\tau_t(\omega)$ by $\inf\{ s \ge 0: \langle \omega \rangle_s \ge t\}$, so considering the hitting time of $[t, \infty)$ rather than $(t,\infty)$. This corresponds to taking the c\`agl\`ad version $(\tau_{t-})_{t \ge 0}$ of the c\`adl\`ag process $(\tau_{t})_{t \ge 0}$. But since on $\Omega^{\mathrm{qv}}$ the paths $\omega$ and $\langle \omega \rangle$ have the same intervals of constancy, we get $\omega(\tau_{t-}) = \omega(\tau_t)$ for all $\omega \in \Omega^{\mathrm{qv}}$, and by Remark~\ref{rem:restriction to qv paths} more generally for all typical price paths in $C(\R_+)$.
\end{remark}

\section{Duality for one period}\label{sec:result}

Here we are interested in a one period duality result for derivatives $G$ on $C_{\mathrm{qv}}[0,1]$ of the form $\omega\mapsto G(\omega,\langle \omega\rangle_1)$ which are invariant under suitable time-changes of $\omega$. Typical examples for such derivatives are the running maximum up to time $1$ or functions of the quadratic variation. Formally, this amounts to 
\begin{equation*}
  G =\tilde G\circ \tntt (\cdot,1)
\end{equation*}
for some optional process $(\tilde G_t)_{t\geq 0}$ on $(\Omega^{\mathrm{qv}},(\F^{\mathrm{qv}}_t)_{t\ge 0})$, and more specifically we will focus on processes $\tilde G$ which are of the form $\tilde G_t(\omega) = \gamma(\omega_{\upharpoonright[0,t]},t)$, where $\omega_{\upharpoonright [0,t]}$ denotes the restriction of $\omega$ to the interval $[0,t]$ and $\gamma\colon \CS \to \R$ is an upper semi-continuous functional which is bounded from above. Here we wrote $\CS$ for the space of {\em stopped paths}
\begin{equation*}
  \CS := \{(f,s)\,:\, f \in C[0,s], s \in \R_+ \},
\end{equation*}
equipped with the distance $d_\CS$ which is defined for $s\leq t$ by
\begin{equation}\label{STop}
  d_\CS \big ((f,s),(g,t)\big ) := \max\Bigg( |t-s|, \sup_{0\le u \le s} |f(u) - g(u)|, \sup_{s \le u \le t} |g(u) - f(s)| \Bigg),
\end{equation}
and which turns $\CS$ into a Polish space. The space $\Upsilon$ is a
convenient way to express optionality of a process on $C(\R_+)$. Indeed, put 
\begin{equation*}
  r\colon C(\R_+)\times \R_+ \to \Upsilon,\quad (\omega,t)\mapsto (\omega_{\upharpoonright [0,t]},t).
\end{equation*}
By \cite[Theorem~IV.~97]{DeMeA}, a process $Y$ is predictable if and
only if there is a function $H\colon\Upsilon\to\R$ such that $Y=H\circ
r$. Moreover, since $\Omega^{qv}$ is a subset of the set of continuous
paths, the optional and predictable processes coincide. Hence $Y$ is
optional if and only if such a function $H$ exists. We can say that an
optional process $Y$ is $\Upsilon$-(upper/lower semi-) continuous if
and only if the corresponding function $H$ on $\Upsilon$ is
(upper/lower semi-) continuous.\medskip

Given a centered probability measure $\mu$ on $\R$ with finite first moment, we want to solve the primal maximization problem 
\begin{equation}\label{eq:primal}
  P := \sup\big\{ \E_\P [G] : \P \text{ is a martingale measure on } C_{\mathrm{qv}}[0,1] \text{ s.t. } S_1(\P) = \mu \big\},
\end{equation}
where $S$ denotes the canonical process on $C_{\mathrm{qv}}[0,1]$.

Since $\mu$ satisfies $\int |x|\dd \mu(x) < \infty$, there exists a smooth convex function $\varphi\colon\R\to\R_+$ with $\varphi(0)=0$, $\lim_{x\to \pm \infty} \varphi(x)/|x| = \infty$, and such that $\int \varphi(x) \dd \mu(x) < \infty$ (apply for example the de~la~Vall\'ee-Poussin theorem). From now on we fix such a function $\varphi$ and we define
\begin{equation*}
  \zeta_t(\omega) := \frac{1}{2} \int_0^t \varphi''(S_s(\omega)) \dd \langle S \rangle_s (\omega), \qquad (\omega,t) \in C_{\mathrm{qv}}[0,1] \times [0,1],
\end{equation*}
where we write $\langle S \rangle(\omega) := \langle \omega \rangle$ for $\omega \in \Omega^{\mathrm{qv}}$. We then consider for $\alpha, c > 0$ the set of (generalized admissible) simple strategies
\begin{equation*}
  \mathcal{Q}_{\alpha,c} :=\big\{H: H \text{ is a simple strategy and } (H \cdot S)_t(\omega) \ge -c - \alpha \zeta_t(\omega) \ \forall (\omega,t) \in C_{\mathrm{qv}}[0,1] \times [0,1] \big\}.
\end{equation*}
We also define the set of ``European options available at price 0'':
\begin{equation*}
  \mathcal{E}^0 := \Big\{\psi \in C(\R): \frac{|\psi|}{1 + \varphi} \text{ is bounded}, \int \psi(x) \dd \mu(x) =0\Big\}.
\end{equation*}

In this setting we shall deduce the following duality result for one period.

\begin{theorem}\label{HedgingDual}
  Let $\gamma\colon \CS \to \R$ be upper semi-continuous and bounded from above and let $\tilde G_t(\omega) = \gamma\circ r(\omega,t)$ and $G(\omega) = \tilde G\circ \tntt(\omega,1)$. Put
  \begin{equation*}
    D := \inf \left\{ p\,:\, \begin{array}{l} \exists c,\alpha>0, (H^n) \subseteq \mathcal{Q}_{\alpha,c}, \psi \in \mathcal{E}^0 \text{ s.t. } \forall \omega \in C_{\mathrm{qv}}[0,1] \\
                           p +\liminf_n (H^n \cdot S)_1(\omega) + \psi(S_1(\omega)) \ge G(\omega)
                         \end{array} \right\},
  \end{equation*}
  then we have the duality relation 
  \begin{equation}\label{HedgingDualEq2}
    P = D.
  \end{equation}
\end{theorem}

The inequality $P \leq D$ is fairly easy: If $p > D$, then there exists a sequence $(H^n) \subseteq \mathcal{Q}_{\alpha,c}$ and a $\psi \in C(\R)$ with $\int \psi(x) \dd \mu(x) = 0$ such that $p +\liminf_n (H^n \cdot S)_1(\omega) + \psi(S_1(\omega)) \ge G(\omega)$. In particular, for all martingale measures $\P$ on $C_{\mathrm{qv}}[0,1]$ with $S_1(\P) = \mu$ we have
\begin{equation*}
  \E_\P[G] \le \E_\P[p +\liminf_{n\to \infty} (H^n \cdot S)_1 + \psi(S_1)] \le p + \liminf_{n\to \infty} \E_\P[(H^n \cdot S)_1] + \E_\P[\psi(S_1)] \le p,
\end{equation*}
where in the second step we used Fatou's lemma, which is justified because $(H^n \cdot S)_1$ is uniformly bounded from below by $-c -\alpha \zeta_1$ and from It\^o's formula we get $\P$-almost surely
\begin{equation*}
  \varphi(S_t) = \int_0^t \varphi'(S_s) \dd S_s + \zeta_t,
\end{equation*}
which shows that $\zeta$ is the compensator of the $\P$-submartingale $\varphi(S)$ and therefore $\E_\P[\zeta_1] < \infty$.\medskip

In the following we concentrate on the inequality $P \geq D$ and proceed in three steps:
\begin{enumerate}
  \item[1.] Reduction of the primal problem~$P$ to optimal Skorokhod embedding~$P^*$: $P=P^*$. 
  \item[2.] Duality of optimal Skorokhod embedding~$P^*$ and a dual problem~$D^*$: $P^*=D^*$.
  \item[3.] The new dual problem~$D^*$ dominates the dual problem~$D$: $D\leq D^*$.
\end{enumerate}

Step~1: The idea, going back to Hobson~\cite{Ho98a}, is to translate the primal problem into an optimal Skorokhod embedding problem.  Let us start by observing that if $\P$ is a martingale measure for $S$, then by the Dambis Dubins-Schwarz theorem the process $(\ntt(S)_{t\wedge \langle S \rangle_1})_{t \ge 0}$ is a stopped Brownian motion under $\P$ in the filtration $(\F^S_{\tau_t})_{t\ge 0}$, where $(\F^S_t)_{t \in [0,1]}$ is the  usual augmentation of the filtration generated by $S$ and where $(\tau_t)_{t\ge 0}$ are the stopping times defined in~\eqref{eq:time-chnage}. It is also straightforward to verify that $\langle S \rangle_1$ is a stopping time with respect to $(\F^S_{\tau_t})$. Since moreover $\ntt(S)_{\langle S \rangle_1} = S_1 \sim \mu$ we deduce that there exists a new filtered probability space $(\tilde{\Omega}, (\G_t)_{t \ge 0}, \Q)$ with a Brownian motion $B$ and a stopping time $\tau$, such that $B_\tau \sim \mu$, the process $B_{\cdot \wedge \tau}$ is a uniformly integrable martingale, and
\begin{equation*}
  \E_\P[G] = \E_\Q[\gamma ((B_s)_{s\leq \tau}, \tau)].
\end{equation*}

Conversely, let $(\tilde{\Omega}, (\G_t)_{t \ge 0}, \Q)$ be a filtered probability space with a Brownian motion $B$ and a finite stopping time $\tau$, such that $B_\tau \sim \mu$ and $B_{\cdot \wedge \tau}$ is a uniformly integrable martingale, and define $(S_t := B_{(t/(1-t)) \wedge \tau})_{t\in[0,1]}$. Then $S$ is a martingale on $[0,1]$ with $\langle S \rangle_1 = \tau$, and writing $\P$ for the law of $S$ we get
\begin{equation*}
  \E_\Q[\gamma ((B_s)_{s\leq \tau}, \tau)] = \E_\P[ \tilde{G}\circ \tntt(S, 1)] = \E_\P[G].
\end{equation*}
To conclude, we arrive at the following observation:

\begin{lemma}
  The value $P$ defined in \eqref{eq:primal} is given by
  \begin{equation}\label{eq:primal2}
    P = P^\ast :=\sup\left\{\E_\Q [\gamma ((B_s)_{s\leq \tau}, \tau)]\,:\, \begin{array}{l}(\tilde{\Omega}, (\G_t)_{t \ge 0}, \Q) \in \mathfrak{F}, \tau \in \mathfrak{T}((\G_t)_{t \ge 0}),  \\B_\tau \sim \mu, B_{\cdot \wedge \tau} \text{ is a u.i. martingale}
    \end{array}\right\},
  \end{equation}
  where $\mathfrak{F}$ denotes all filtered probability spaces supporting a Brownian motion $B$ and\\
  $\mathfrak{T}((\G_t)_{t\ge 0})$ is the set of $(\G_t)_{t \ge 0}$-stopping times.
\end{lemma}

By \cite[Lemma~3.11]{BeCoHu14}, the value $P^\ast$ is independent of the particular probability space as long as it supports a Brownian motion and a $\G_0$-measurable uniformly distributed random variable. Therefore, it is sufficient to consider the probability space $(\ol{\Omega},\ol{\F},(\ol{\F}_t)_{t \ge 0},\ol{\W})$, where we take $\ol{\Omega} := C(\R_+) \times [0,1]$, $\F=(\F_t)_{t\geq 0}$ to be the natural filtration on $C(\R_+)$, $\ol{\F}$ to be the completion of  $\F \otimes \mathcal{B}([0,1]), \ol{\W}(A_1\times A_2): = \W(A_1) \Leb(A_2)$, and  $\ol{\F}_t$ the usual augmentation of $\F_t \otimes \sigma([0,1])$. Here, $\Leb$ denotes the Lebesgue measure and $\W$ the Wiener measure. We will write $\bar B=(\bar B_t)_{0\leq t}$ for the canonical process on $\bar \Omega$, that is $\bar B_t(\omega,u):=\omega(t)$. 

Given random times $\tau, \tau'$ on $\ol \Omega$ and a bounded continuous function $f\colon C(\R_+)\times \R_+\to \R$ we define 
\begin{equation*}
  d_f(\tau,\tau'):= \left| \E_{\ol\W} [ f(\tau)-f(\tau')]\right| = \left| \int  [ f(\omega, \tau(\omega, x))-f(\omega, \tau'(\omega, x))] \ol\W(\mathrm{d} \omega, \mathrm{d} x) \right|.
\end{equation*}
We then identify $\tau$ and $\tau'$ if $d_f(\tau,\tau')=0$ for all
continuous and bounded $f$. On the resulting space of equivalence classes denoted by $\mathsf{RT}$, the family of semi-norms $(d_f)_f$ gives rise to a Polish topology. An equivalent interpretation of this space is to consider the measures on $C(\R_+)\times \R_+$ induced by 
\begin{equation}\label{eq:taumeasure}
  \nu_\tau (A\times B) = % \E_{\ol\W}\left[\1_{\omega \in A, \tau(\omega,x) \in B}\right] =
  \int \1_{\omega \in A, \tau(\omega,x) \in B} \ol\W(\mathrm{d} \omega, \mathrm{d} x).
\end{equation}
The topology above corresponds to the topology of weak convergence of the corresponding measures. A random time $\tau$ is a $\ol\F$ stopping time if and only if for any $f\in C(\R_+)$ supported on $[0,t]$ the random variable $f(\tau)$ is $\ol\F_t$ measurable which in turn holds if and only if for all $g\in C_b(C(\R_+))$ we have (see also \cite[Theorem~3.8]{BeCoHu14})
\begin{equation}\label{eq:STchar}
  \E_{\ol\W}[f(\tau)(g-\E_{\W}[g|\F_t])]=\int f(s)(g-\E_{\W}[g|\F_t])(\omega)~\nu_\tau(\d\omega,\d s)=0,
\end{equation}
where on the left hand side we interpret $g-\E_{\W}[g|\F_t]$ as a random variable on the extension $\ol\Omega$ via $(g-\E_{\W}[g|\F_t])(\omega,x)=(g-\E_{\W}[g|\F_t])(\omega)$. As a consequence, for a stopping time~$\tau$ on $\ol \Omega$ all elements of the respective equivalence class are stopping times. We will call this equivalence class, as well as  (by abuse of notation) its representatives \emph{randomized stopping times} (in formula: $\RST$).

By the same argument as above, there exists a continuous compensating process $\zeta^1 \colon\S\to \R$ such that $(\phi(B_t) - \zeta_t)$ is a martingale under $\W$. We write $\RST(\mu)$ for the set of randomized stopping times which embed a given measure $\mu$ (that is $\bar B_\tau \sim \mu$ and $B_{\cdot \wedge \tau}$ is a uniformly integrable martingale), and such that $\E_{\ol\W}[\zeta^1_\tau]<\infty$, this last condition also being equivalent to $\E_{\ol\W}[\zeta^1_\tau] = V$ for $V = \int\phi(x) \,\mu(\d x)$. It it is then not hard to show that $\RST(\mu)$ is compact, see \cite[Theorem~3.14, and Section~7.2.1]{BeCoHu14}. Thereby, we have turned the optimization problem~\eqref{eq:primal} into the primal problem of the optimal Skorokhod embedding problem
\begin{equation}\label{eq:primalSEP}
  P^\ast=\sup_{\tau\in\RST(\mu)}\E_{\ol\W}[\gamma((\bar B_s)_{s\leq \tau}, \tau)].
\end{equation}

Step~2:
In \cite{BeCoHu14} a duality result for \eqref{eq:primalSEP} is shown. To state it (and in what follows), it will be convenient to fix a particularly nice version of the conditional expectation on the Wiener space $(C(\R_+),\F,\W)$.

\begin{definition}\label{EAverage}
  Let $X\colon C(\R_+)\to \R$ be a measurable function which is bounded or positive. Then we define $\E_\W[X|\F_t]$ to be the unique $\F_t$-measurable function satisfying 
  \begin{equation*}
    \E_\W[X|\F_t](\omega) := \textstyle \int X((\omega_{\llcorner[0,t]})\oplus \tilde \omega)\, \W(\dds \tilde \omega),
  \end{equation*} 
  where $(\omega_{\upharpoonright [0,t]}) \oplus \tilde \omega$ is the concatenation of $\omega_{\upharpoonright [0,t]}$ and $\tilde \omega$, that is $(\omega_{\upharpoonright [0,t]}) \oplus \tilde \omega(r) := \1_{r \le t} \omega(r) + \1_{r> t} (\omega(t)+\tilde\omega(r-t))$. Similarly, for bounded or positive $X\colon\Omega^\mathrm{qv}\to\R$ we define $\E_\W[X|\F^\mathrm{qv}_t]$ to be the unique $\F^\mathrm{qv}_t$-measurable function satisfying 
  \begin{equation*}
    \E_\W[X|\F^\mathrm{qv}_t](\omega)= \textstyle \int X((\omega_{\llcorner[0,t]})\oplus \tilde \omega)\, \W(\dds \tilde \omega).
  \end{equation*}
  Then $\E_\W[X|\F_t](\omega)$ as well as $\E_\W[X|\F^\mathrm{qv}_t](\omega)$ depend only on $\omega_{\upharpoonright[0,t]}$, and in particular we can (and will) interpret the conditional expectation also as a function on $C_{\mathrm{qv}}[0,t] := \{ \omega_{\upharpoonright[0,t]}\,:\, \omega \in \Omega^{\mathrm{qv}}\}$.
\end{definition}

We equip $\Omega^{\mathrm{qv}}$ with the topology of uniform convergence on compacts. Note that then $\Omega^{\mathrm{qv}}$ is a metric space, but it is not complete due to the fact that it is possible to approximate paths without quadratic variation uniformly by typical Brownian sample paths.

\begin{proposition}[{\cite[Proposition 3.5]{BeCoHu14}}]\label{prop:very cont}
  Let $X\in C_b(C(\R_+))$. Then $X_t(\omega):=\E_\W[X|\mathcal F_t](\omega)$ defines a $\CS$-continuous martingale on $(C(\R_+),(\F_t), \W)$. By restriction, it is also a $\CS$-continuous martingale on $(\Omega^\mathrm{qv},(\F^\mathrm{qv}_t), \W)$ % We denote this martingale by $X^M.$
\end{proposition}

Then the duality for the optimal Skorokhod embedding reads:

\begin{theorem}\label{thm:SEPDual}
  Let $\gamma\colon \CS \to \R$ be upper semi-continuous and bounded from above. We put
  \begin{equation*}
    D^\ast :=\inf\left\{ p:
    \begin{array}{l}
    \exists \alpha \ge 0, \psi \in \mathcal{E}^0, m \in C_b(C(\R_+)) \text{ s.t. } \E_\W[m] = 0 \text{ and } \forall (\omega, t)\in C(\R_+)\times \R_+ \\
    p + \E_\W[m|\F_t](\omega) + \alpha Q(\omega,t) + \psi(B_t(\omega)) \geq \gamma(\omega,t)
    \end{array}
    \right\},
  \end{equation*}
  where we wrote $Q(\omega,t):=\varphi(B_t(\omega)) - 1/2 \int_0^t \varphi''(B_s(\omega)) \dd s$. 
  Let $P^\ast$ be as defined in~\eqref{eq:primal2}. Then one has
  \begin{equation*} 
    P^\ast = D^\ast.
  \end{equation*}
\end{theorem}

\begin{proof}
  This is essentially a restatement of \cite[Theorem~4.2 \& Proposition~4.3 (c.f. Proof of Theorem~4.2)]{BeCoHu14}, combined with the discussion before \cite[Theorem~7.3]{BeCoHu14}, which enables us to modify the statement to include the term $\alpha Q(\omega,t)$ instead of $\alpha(\omega(t)^2-t/2)$.
\end{proof}

By Proposition~\ref{prop:very cont} and the fact that $\Omega^{\mathrm{qv}}$ is dense in $C(\R_+)$, we see that the value $D^\ast$ equals
\begin{equation*}
  D^{\ast,\mathrm{qv}} :=\inf\left\{ p:
  \begin{array}{l}
    \exists \alpha \ge 0, \psi \in \mathcal{E}^0, m \in C_b(\Omega^{\mathrm{qv}}) \text{ s.t. } \E_\W[m] = 0 \text{ and } \forall (\omega, t)\in \Omega^{\mathrm{qv}}\times \R_+ \\
    p + \E_\W[m|\F_t^{\mathrm{qv}}](\omega) + \alpha Q(\omega,t) + \psi(B_t(\omega)) \geq \gamma(\omega,t)
  \end{array}
  \right\}.
\end{equation*}

Step~3: Let now $p>D*=P^\ast=P$. Then Theorem~\ref{thm:SEPDual} gives us a function $\psi \in \mathcal{E}^0$, a constant $\alpha \geq 0$, and a continuous bounded function $m\colon \Omega^{\mathrm{qv}} \to \R$ with $\E_\W[m]=0$ such that for all $(\omega,t) \in \Omega^{\mathrm{qv}} \times \R_+$
\begin{equation}\label{eq:martingale-ineq}
  M_t(\omega) := \E_\W[m|\F^{\mathrm{qv}}_t](\omega) \geq -p-\psi(B_t(\omega)) - \alpha Q(\omega,t) + \gamma(\omega,t).
\end{equation}
Consider now the functional $\tilde m \colon \Omega^{\mathrm{qv}} \to \R$ given by
\begin{equation*}
  \tilde m:= m \circ \ntt
\end{equation*}
which is $\G$-invariant, i.e.\ invariant under all strictly increasing unbounded time-changes, and satisfies $\E_\W[\tilde m]=\E_\W[m]=0$. Denote by $m_0$ the supremum of $|m(\omega)|$ over all $\omega \in \Omega^{\mathrm{qv}}$. Then $m_0+m \ge 0$, and if we fix $\varepsilon >0$ and apply Theorem~\ref{thm:vovk-dds} in conjunction with Remark~\ref{rmk:outer}, we obtain a sequence of simple strategies $(\tilde H^n)\subseteq \mathcal{H}_{m_0+\varepsilon}$ such that
\begin{equation*}
  \liminf_{t \to \infty} \liminf_{n\to \infty} \varepsilon+(\tilde H^n\cdot B)_t(\omega) \geq \tilde m(\omega) \1_{\{\langle B \rangle_\infty = \infty\}}(\omega), \qquad \omega \in \Omega^{\mathrm{qv}}.
\end{equation*}
By stopping we may suppose that $(\tilde H^n\cdot B)_t(\omega) \le m_0$ for all $(\omega,t) \in \Omega^{\mathrm{qv}} \times \R_+$. Set 
\begin{equation*}
  \tilde M_t(\omega):= (M\circ \tntt)(\omega,t), 
   \qquad (\omega,t) \in \Omega^{\mathrm{qv}}\times \R_+ .
\end{equation*}

\begin{lemma}\label{lem:cts domination}
 For all $(\omega,t) \in \Omega^{\mathrm{qv}}\times \R_+$ we have
  \begin{equation*}
    \varepsilon + \liminf_{n\to \infty}(\tilde H^n \cdot B)_t(\omega) \geq  \tilde M_t(\omega).%,
  \end{equation*}
  %where we recall that $B$ denotes the canonical process on $C(\R_+)$.
\end{lemma}

\begin{proof}
  We claim that $\tilde M_t = \E_\W[ \1_{\{\langle B \rangle_\infty = \infty\}} \tilde m |\F^{\mathrm{qv}}_t]$. Indeed we have 
  \begin{equation*}
    \tilde M_t\big( \omega_{\upharpoonright [0,t]}\oplus \tilde \omega, t\big) = (M\circ \tntt) \big( \omega_{\upharpoonright [0,t]}\oplus \tilde \omega, t\big)= M_{\langle B\rangle_t}\big(\ntt( \omega_{\upharpoonright [0,t]}\oplus \tilde \omega)\big),
  \end{equation*}
  where the latter quantity actually does not depend on $\tilde \omega$, i.e.\ with a slight abuse of notation we may write it as $M_{\langle B \rangle_t} \big(\ntt(\omega_{\upharpoonright[0,t]})\big)$. Also, we have
  \begin{align*}
    \E_\W[ \1_{\{\langle B \rangle_\infty = \infty\}} \tilde m |\F^{\mathrm{qv}}_t](\omega_{\upharpoonright[0,t]}) & = \E_\W[ \1_{\{\langle B \rangle_\infty = \infty\}} m\circ \ntt|\F^{\mathrm{qv}}_t](\omega_{\upharpoonright[0,t]})\\
    &= \int  \1_{\{ \langle B \rangle_\infty = \infty\}}(\omega_{\upharpoonright[0,t]}\oplus \tilde \omega) (m\circ\ntt) (\omega_{\upharpoonright[0,t]}\oplus \tilde \omega)\, \W(\dds \tilde \omega)\\
    &= \int \1_{\{\langle B \rangle_\infty = \infty\}}(\tilde \omega) m\big(\ntt (\omega_{\upharpoonright[0,t]})\oplus \ntt(\tilde \omega) \big)\, \W(\dds \tilde \omega)\\
    & = \int m\big(\ntt (\omega_{\upharpoonright[0,t]})\oplus \tilde\omega \big)\, \W(\dds\tilde\omega) \\
    & = M_{\langle B\rangle_t}\big(\ntt(\omega_{\upharpoonright[0,t]})\big),
  \end{align*}
  where we used that $\W$-almost surely $\tilde{\omega} = \ntt(\tilde \omega)$ and $\langle B \rangle_\infty = \infty$. Writing $(\tilde H^n \cdot B)_t^s := (\tilde H^n \cdot B)_s - (\tilde H^n \cdot B)_t$, we thus find
  \begin{align*}
    \tilde M_t & = \E_\W[\1_{\{\langle B \rangle_\infty = \infty\}} \tilde m|\F^{\mathrm{qv}}_t] \leq \varepsilon + \E_\W[  \liminf_{s\to \infty} \liminf_{n\to \infty} (\tilde H^n \cdot B)_s|\F^{\mathrm{qv}}_t] \\
    & = \varepsilon + \E_\W[\liminf_{s\to \infty} \liminf_{n\to \infty} ((\tilde H^n \cdot B)_t + (\tilde H^n \cdot B)_t^s) |\F^{\mathrm{qv}}_t]\\
    &= \varepsilon + \liminf_{n\to \infty} (\tilde H^n \cdot B)_t + \E_\W[ \liminf_{s\to \infty} \liminf_{n\to \infty} (\tilde H^n \cdot B)_t^s|\F^{\mathrm{qv}}_t].
  \end{align*}
  Now it is easily verified that $(\liminf_n (\tilde H^n \cdot B)_t^s)_{s \ge t}$ is a bounded $\W$-supermartingale started in 0 (recall that $-m_0 - \varepsilon \le (\tilde H^n\cdot B)_s(\omega) \le m_0$ for all $(\omega,s) \in \Omega^{\mathrm{qv}} \times \R_+$, which yields $|(\tilde H^n \cdot B)_t^s(\omega)| \le 2 m_0 + \varepsilon$ for all $(\omega,s) \in \Omega^{\mathrm{qv}} \times \R_+$), and therefore the conditional expectation on the right hand side is nonpositive, which concludes the proof.
\end{proof}

We are now ready to show that $D\leq D^*$ and thus to prove the main result (Theorem~\ref{HedgingDual}) of this section.

\begin{proof}[Proof of Theorem~\ref{HedgingDual}]
  Lemma~\ref{lem:cts domination} and \eqref{eq:martingale-ineq} show that
  \begin{equation*}
    \varepsilon + \liminf_{n \to \infty} (\tilde H^n \cdot
    B)_t(\omega) \geq -p -\psi((B\circ \tntt)(\omega,t))- \alpha (Q\circ \tntt) (\omega,t) + \gamma\circ\tntt(\omega,t)
  \end{equation*}
  for all $(\omega,t) \in \Omega^{\mathrm{qv}} \times \R_+$. Noting
  that $\psi((B \circ \tntt) (\omega,t))=\psi(B_t(\omega))$ and $Q\circ\tntt(\omega,t)= \varphi(B_t(\omega)) - \zeta_t(\omega)$, we get
  \begin{align*}
    & p+ \varepsilon + \liminf_{n \to \infty} (\tilde H^n \cdot B)_t(\omega) + \psi(B_t(\omega)) + \alpha(\varphi(B_t(\omega)) - \zeta_t(\omega))\big)  \ge \gamma\circ \tntt(\omega,t)
  \end{align*}
  for all $(\omega, t) \in \Omega^{\mathrm{qv}} \times \R_+$. 
  %  \cite[Theorem 5.1]{Vo12} shows that the complement of the set in the indicator function on the right hand side has outer measure $0$, and therefore we obtain a new sequence of simple strategies $(G^n) \subset \mathcal{H}_{m_0+2\varepsilon}$ such that 
  %  \begin{equation*}
  %    p+ 2\varepsilon + \liminf_n (G^n \cdot B)_t(\omega) + \psi(B_t(\omega))+\alpha(\varphi(B_t(\omega)) - \zeta_t(\omega)) \ge \gamma\circ \tntt(\omega,t)
  %  \end{equation*}
  %  for all $(\omega,t) \in \Omega \times \R_+$.
  It now suffices to apply F\"ollmer's pathwise It\^o formula~\cite{Fo81} along the dyadic Lebesgue partition defined in Section~\ref{sec:setting} to obtain a sequence of simple strategies $(G^n)\subseteq \mathcal{Q}_{1,\alpha}$ such that $\lim_{n \to \infty} (G^n\cdot B)_t(\omega) = \alpha(\varphi(B_t(\omega)) - \zeta_t(\omega))$ for all $(\omega,t) \in \Omega^{\mathrm{qv}} \times \R_+$; to make the strategies $(G^n)$ admissible it suffices to stop once the wealth at time $t$ drops below $-1-\alpha\zeta_t(\omega) < \alpha(\varphi(B_t(\omega)) - \zeta_t(\omega))$. Hence, setting $H^n := \tilde{H}^n + G^n$, we have established that there exist $(H^n)\subseteq \mathcal{Q}_{m_0+\varepsilon+1,\alpha}$ and $\psi \in \mathcal{E}^0$ such that 
  \begin{equation*}
    p+\varepsilon + \liminf_{n \to \infty} (H^n\cdot B)_t(\omega)+\psi(B_t(\omega)) \geq \gamma\circ\tntt(\omega,t)
  \end{equation*}
  for all $(\omega,t) \in \Omega^{\mathrm{qv}} \times \R_+$. Now for fixed $t \in \R_+$ the functionals on both sides only depend on $\omega_{\upharpoonright[0,t]}$, 
  % are adapted (see Lemma~\ref{lem:tntt predictable}), so for fixed time $t$ 
  so we can consider them as functionals on $C_{\mathrm{qv}}[0,t]$, and thus the inequality holds in particular for all $(\omega,t) \in C_{\mathrm{qv}}[0,1]\times[0,1]$.
  Since $p>P$ and $\varepsilon>0$ are arbitrarily small, we deduce that $D \le P$ and thus that $D = P$.
\end{proof}

%\begin{remark}
%  We observe that the requirement $G = \tilde G\circ \tntt(\cdot,1)$ in Theorem~\ref{HedgingDual} can now easily be weakened to require only $G = \tilde G\circ \tntt(\cdot,1)$ outside of a $\overline{P}$-null set. Indeed, by Lemma~\ref{lem:na1} we can find a sequence of simple strategies with arbitrarily small cost which superhedge the payoff on the set where $G \neq \tilde G\circ \tntt(\cdot,1)$. One particular case where this difference is important is on the class of paths with smooth sections. If a path contains an interval on which the path is smooth, starts and ends at the same point, and is not constant, then the normalising time transformation will simply cut out this section of the path. In some cases (e.g. when the payoff depends on the running maximum), then $G = \tilde G\circ \tntt(\cdot,1)$ may not hold for such paths. However, it is easily checked that these paths form an atypical set, and hence may be ignored in computing the super-hedging price.
%\end{remark}

\section{Duality in the multi-marginal case}\label{sec:duality}

In this section, we will show a general duality result for the multi-marginal Skorokhod embedding problem and moreover, for a slightly more general problem. Our main result will then follow by exactly the same steps and arguments as for the one marginal duality, that is reduction of the primal problem to optimal multi-marginal Skorokhod embedding (Step 1 in the last section) and domination of the dual problem via the dual in the optimal multi-marginal Skorokhod embedding (Step 3 in the last section).\medskip

To this end, we introduce the set of all randomized multi stopping times or $n$-tuples of randomized stopping times.  As before we consider the space $(\bar\Omega,\bar\F,\bar \W)$ and denote its elements by $(\omega,x)$. We consider all $n$-tuples $\tau=(\tau_1,\ldots,\tau_n)$ with $\tau_1\leq\ldots\leq\tau_n$ and $\tau_i\in \mathsf{RT}$ for all $i$. We identify two such tuples if
\begin{align}\label{EssEqual}
  d_f( \tau,  \tau') & := \left|\E_{\bar\W} [f(\tau_1,\ldots, \tau_n) - f(\tau'_1,\ldots, \tau'_n)]\right| \\  \nonumber
  & \,=  \left| \int [f(\omega, \tau_1(\omega,x),\ldots, \tau_n(\omega,x)) -  f(\omega, \tau'_1(\omega,x),\ldots, \tau'_n(\omega,x))] \ol\W(\mathrm{d}\omega, \mathrm{d}x)\right| 
\end{align} vanishes 
for all continuous,  bounded $f\colon C(\R_+)\times \R_+^n\to\R$ and denote the resulting space by $\mathsf{RT}_n$. Moreover, we consider $\mathsf{RT}_n$ as a topological space by testing against all continuous bounded functions as in \eqref{EssEqual}. As for the one-marginal case, for an ordered tuple $\tau_1\leq\ldots\leq \tau_n$ of stopping times it follows from \eqref{eq:STchar} that all elements of the respective equivalence class are ordered tuples of stopping times as well. We will denote this class by $\RST_n$.

% \begin{proposition}\label{prop:RSTn}
%   Let $k \le n$. 
%   \begin{enumerate}
%     \item If $\tau_1 \le \dots \le \tau_{k}$ are stopping times, and
%       $\tau_{k+1}\leq \dots \leq \tau_n$ are pseudo stopping times on $\bar\Omega$ with respect to\ $\bar \F$,
%       and $\tau_1'\leq \ldots\leq \tau_n'$ get identified with them by \eqref{EssEqual}, then
%       $\tau_1', \dots \tau_{k}'$ are stopping times as well, and $\tau_{k+1}',\dots,\tau_n'$ are
%       pseudo-stopping times. We then say that $ (\tau_1, \ldots, \tau_n)$ is a $(k,n)$-pseudo
%       multi-stopping time. If $k=n$, we simply say that $ (\tau_1, \ldots, \tau_n)$ is a
%       multi-stopping time. We write $\PRST_{k,n}$ for the set of $(k,n)$-pseudo multi-stopping times,
%       and $\RST_n = \PRST_{n,n}$.
%     \item The set $\PRST_{k,n}$ of all $(k,n)$-pseudo multi-stopping times $(\tau_1\leq \ldots \leq \tau_n)$ is closed for all $k \le n$. In particular, the set $\RST_n$ of multi-stopping times is closed.
%   \end{enumerate}
% \end{proposition}
% 
% Below we will only be interested in the sets $\PRST_{k,n}$ for $k=n$, and $k=n-1$.

Fix $I\subset \{1,\ldots,n\}$ with $n\in I$ and $|I|\leq n$ measures $(\mu_i)_{i\in I}=\mu$ in convex order with finite first moment. If $i \in \{1,\dots,n\}\setminus I$, write $i+$ for the smallest element of $\{j \in I: j \ge i\}$. For $i\in I$ we set $i+=i$. By an iterative application of the de~la~Vall\'ee-Poussin Theorem, there is an increasing family of smooth, non-negative, strictly convex functions $(\phi_i)_{i=1,\dots,n}$ (increasing in the sense that $\phi_i\leq \phi_j$ for $i\leq j$) such that $\phi_i(0)=0$ and $\phi_{i+1}/\phi_i \to \infty$ as $x \to \pm \infty$, and $\int \phi_i \dd \mu_{i+}<\infty$ for all $i=1,\dots,n$. Denote the corresponding compensating processes by $\zeta^i$ such that $Q^{i}:=\phi_i( B)-\zeta^i$ is a martingale.  We also write $\cE_i := \left\{\psi \in C(\R): \frac{|\psi|}{1 + \varphi_i} \text{ is bounded}\right\}$.

Then, we define $\RST_n(\mu)$  to be the subset of $\RST_n$ consisting of all tuples $(\tau_1\leq\ldots\leq\tau_n)$ such that $\bar B_{\tau_i}\sim \mu_i$ for all $i\in I$ and $\E_{\ol\W}[\zeta^n_{\tau_n}] < \infty$. Similar to the one-marginal case we get

\begin{lemma}
  For any $I\subset \{1,\ldots,n\}$ with $n\in I$ and any family of measures $(\mu_i)_{i\in I}=\mu$ in convex order the set $\RST_n(\mu)$  is compact.
\end{lemma}

We introduce the space of paths where we have stopped $n$ times:
\begin{equation*}
  \S_n:=\big\{(f,s_1,\ldots,s_n)\,:\,(f,s_n)\in \S, 0\leq s_1\leq \ldots \leq s_n\big\},
\end{equation*}
equipped with the topology generated by the obvious analogue of
\eqref{STop}:
\begin{equation*}
  d_{\S_n}\left( (f,s_1, \ldots, s_n),(g,t_1,\dots,t_n)\right) = \max
  \left( |s_1-t_1|, \ldots, |s_n-t_n|, \sup_{u\ge 0}
    |f(u\wedge s_n)-g(u\wedge t_n)|\right).
\end{equation*}
We put $\Delta_n:=\{(s_1,\ldots,s_n)\in\R_+^n\colon s_1\leq\ldots\leq s_n\}$. As a natural extension of an optional process, we say that a process $Y\colon C(\R_+) \times \Delta_n$ is optional if for any family of stopping times $\tau_1 \le\dots \le \tau_n$, the map $Y(\bar B,\tau_1,\dots,\tau_n)$ is $\ol\F_{\tau_n}$-measurable. Put
\begin{equation*}
 r_n\colon C(\R_+)\times \Delta_n\to\S_n, \quad (\omega,s_1,\ldots,s_n)\mapsto (\omega_{\upharpoonright [0,s_n]},s_1,\ldots,s_n).
\end{equation*}
Just as in the one-marginal case a function $Y\colon C(\R_+)\times\Delta_n\to\R$ is optional if and only if there exist a Borel function $H\colon\S_n\to\R$ such that $Y=H\circ r_n$. \medskip

% On $\S_n$, the optional processes are functions $\gamma\colon\S_n \to \R$. Indeed, it is easy to show that a Borel function, $\gamma\colon \S_n \to \R$ is an optional process in the sense given above, and moreover, any optional process can be written in this way (see \cite{BeCoHu14} for the case with one stopping time; the general case is then immediate).

Given $\gamma\colon \S_n\to\R$, we are interested in the following $n$-step primal problem
\begin{align}\label{primal}
  P^*_n:=\sup\left\{ \E_{\ol\W}[\gamma\circ r_n(\omega,\tau_1,\ldots,\tau_n)]: (\tau_i)_{i=1}^n \in\RST_n(\mu)\right\}
\end{align}
and its relation to the dual problem
\begin{align}\label{dual}
  D^*_n:=\inf\left\{a: \begin{array}{l} \text{ there exist } (\psi_j)_{j\in I}, \text{ martingales }(M^i)_{i=1}^n, \E_{\W}[M^i_\infty]=0, \int \psi_j \dd\mu_j=0,\\
		     a+\sum_{j\in I} \psi_j(B_{t_j}(\omega)) + \sum_{i=1}^n M^i_{t_i}(\omega) \geq \gamma(\omega,t_1,\ldots,t_n)\\ \text{for all } \omega\in C(\R_+),(t_1,\ldots,t_n)\in\Delta_n
		   \end{array}
         \right\}.
\end{align}

\begin{remark}
  Note that in the primal as well as dual problem only the stopping times truly live on $\bar\Omega$. The martingales $M^i$ as well as the compensators $\zeta^i$ live on $C(\R_+)\times \R_+$ in that they satisfy e.g.\ $M^i_t(\omega,x)=M^i_t(\omega).$ We stress this by suppressing the $x$ variable and writing e.g.\ $\E_\W[M^i_\infty]=0$ rather than $\E_{\ol \W}[M^i_\infty]=0$.
\end{remark}

\begin{convention}\label{rem:IC} 
  In the formulation of $D^*_n$ in \eqref{dual} and in the rest of this section $M^1, \ldots, M^{n}$ will range over $\S$-continuous martingales such that $M^i_t(\omega)=\E_{\ol\W}[m^i|\F_t^0](\omega)+Q_t(\omega)$ for some $m^i \in C_b(\Omega)$ and $Q_t(\omega)=f(B_t(\omega))-\zeta^f_t(\omega)$ where $f$ is a smooth function such that $|f|/(1+\phi_i)$ is bounded, and $\zeta^f$ is the corresponding compensating process $\zeta^f = \frac{1}{2} \int_0^\cdot f''(B_s) \dd s$.
  In addition, we assume that  $\psi_i\in \cE_i$ for all $i \le n$.
\end{convention} 

\begin{theorem}\label{thm:duality}
  Let $\gamma\colon\S_n\to\R$ be upper semicontinuous and bounded from above. Under the above assumptions we have $P_n^*=D_n^*$.
\end{theorem}

As usual the inequality $P_n^*\leq D_n^*$ is not hard to see. The proof of the opposite inequality is based on the following minmax theorem. 

\begin{theorem}[{see e.g.\ \cite[Theorem~45.8]{St85}  or \cite[Theorem~2.4.1]{AdHe96}}]\label{minmax}
  Let $K, L$ be convex subsets of vector spaces $H_1$ respectively $H_2$, where $H_1$ is locally convex and let $F\colon K\times L\to\R$ be given. If
  \begin{enumerate}
    \item K is compact,
    \item $F(\cdot, y)$ is continuous and convex on $K$ for every $y\in L$,
    \item $F(x,\cdot)$ is concave on $L$ for every $x\in K$,
  \end{enumerate}
  then
  \begin{equation*}
    \sup_{y\in L}\inf_{x\in K} F(x,y)=\inf_{x\in K}\sup_{y\in L} F(x,y).
  \end{equation*}
\end{theorem}

The inequality $P_n^* \ge D_n*$ will be proved inductively on $n$. To this end, we need the following preliminary result.

\begin{theorem}\label{thm:auxDual}
  Let $c\colon C(\R_+)\times \Delta_2\to\R$ be upper semicontinuous and bounded from above and let $V_i=\int\phi_i \dd\mu_i<\infty$ for $i=1,2$. Put 
  \begin{equation*}
    P^{V_2}:=\sup\big\{ \E_{\ol\W}[c(\omega,\tau_1,\tau_2)]\,:\, \tau_1\in \RST_{1}(\mu_1), \E_{\ol\W}[\zeta^2_{\tau_2}]\leq V_2,(\tau_1,\tau_2)\in\mathsf{RT}_2\big\}
  \end{equation*}
  and
  \begin{equation*}
    D^{V_2}:=\inf\left\{\int\psi_1 \dd\mu_1 : 
    \begin{array}{l}
       m\in C_b(C(\R_+)), \psi_1 \in C_b(\R_+) , \E_\W[m]=0, \exists \alpha_1,\alpha_2\geq 0\\
       m(\omega)+\psi_1(\omega(t_1)) - \sum_{i=1}^2\alpha_i(V_i-\zeta^i_{t_i}(\omega))\geq c(\omega,t_1,t_2)
    \end{array}\right\}.
  \end{equation*}
  Then, we have
  \begin{equation*}
    P^{V_2}=D^{V_2}.
  \end{equation*}
\end{theorem}

\begin{proof}
  The inequality $P^{V_2}\leq D^{V_2}$ follows easily. We are left to show the other inequality. The idea of the proof is to use a variational approach  together with Theorem~\ref{minmax} to reduce the claim to the classical duality result in optimal transport. 

  Using standard approximation procedures (see\ \cite[Proof of Theorem~5.10~(i), step~5]{Vi09}), we can assume that $c$ is continuous and bounded,  bounded from above by 0 and satisfies for some $L$
  $$\mbox{supp}(c)\subset C(\R_+)\times [0,L]^2.$$
  In the following, we want to apply Theorem~\ref{minmax} where we take for $K$ certain subsets of $\RT_2$. The convexity of these subsets is easily seen by interpreting elements of these sets as measures via the obvious extension of \eqref{eq:taumeasure}. Compactness follows by Prokhorov's Theorem: this is shown by a trivial modification of the argument in \cite[Theorem~3.14]{BeCoHu14}).

  Hence, it follows using Theorem~\ref{minmax} that
  \begin{align*}
    \sup_{ \substack{\tau_1\in \RST_{1}(\mu_1)\\ \E_{\ol\W}[\zeta_{\tau_2}^2]\leq V_2\\ (\tau_1,\tau_2)\in\mathsf{RT}_2}} \E_{\ol\W}[c(\omega,\tau_1,\tau_2)]
    & = \sup_{ \substack{\tau_1\in \RST_{1}(\mu_1)\\ \tau_2\leq \max\{L,\tau_1\}\\ (\tau_1,\tau_2)\in\mathsf{RT}_2}} \ \inf_{\alpha\geq 0}\ \E_{\ol\W}[c(\omega,\tau_1,\tau_2) + \alpha(V_2-\zeta_{\tau_2}^2(\omega))]\\
    & = \inf_{\alpha \geq 0}\ \sup_{  \substack{\tau_1\in \RST_{1}(\mu_1)\\ \tau_2\leq \max\{L,\tau_1\}\\(\tau_1,\tau_2)\in\mathsf{RT}_2}} \ \E_{\ol\W}[c(\omega,\tau_1,\tau_2) + \alpha(V_2-\zeta_{\tau_2}^2(\omega))]\\
    & =  \inf_{\alpha \geq 0} \ \sup_{\tau_1\in\RST_1(\mu_1)} \ \E_{\ol\W}[ c_{\alpha}(\omega,\tau_1)],
  \end{align*}
  where 
  \begin{equation*} 
     c_{\alpha}(\omega,t_1)=\sup_{t_1\leq t_2\leq \max\{L,t_1\}}\ c(\omega,t_1,t_2) + \alpha(V_2-\zeta_{t_2}^2(\omega)).
  \end{equation*}
  %    {\color{red} Note that the convexity of the corresponding subset of
  %      $\RT(\tau_1,\tau_2)$ is easily seen by interpreting elements of
  %      this set as measures via the obvious extension of
  %      \eqref{eq:taumeasure}. Compactness of this set follows by
  %      Prokorov's Theorem: this is easily seen by a trivial modification
  %      of the argument in \cite[Theorem~3.14]{BeCoHu14}).}

  Hence, $c_{\alpha}$ is a continuous and bounded function on $C(\R_+)\times\R_+$ since $c$ is bounded, $\zeta^2$ is continuous and increasing, and $\{t_2:t_1\leq t_2\leq \max\{L,t_1\}\}$ is closed. To move closer to a classical transport setup we define $F\colon C(\R_+)\times\R_+\times\R\to [-\infty,0]$ by 
  $$ F(\omega,t,y):=\begin{cases}
                   c_\alpha(\omega,t) & \text{ if } \omega(t)=y \\
		  -\infty & \text{ else }
                  \end{cases},$$
  which is an upper semicontinuous and bounded function supported on
  $C(\R_+)\times [0,L]\times\R$. Moreover, we define
  $\TRST(\mu_1)$ to consist of all pairs of random variables  $(\tau, Y)$ on $(\ol\Omega,\ol\W)$ such that $Y\sim \mu_1$ and $\tau\in\RST$ satisfies $\E_{\ol\W}[\zeta^1_{\tau}]<\infty.$ If $\tau_1\in\RST(\mu_1)$, then $(\tau_1, \bar B_{\tau_1})\in\TRST(\mu_1)$ and 
  $$ \E_{\ol\W}[c_\alpha(\omega,\tau_1)] = \E_{\ol\W}[F(\omega,\tau_1,\bar B_{\tau_1})]>-\infty.$$
  Conversely, if $(\tau,Y)\in\TRST(\mu_1)$ with $\E_{\ol\W}[F(\omega,\tau,Y)]>-\infty$ almost surely $Y=B_\tau\sim \mu_1$ so that $\tau\in\RST(\mu_1).$ Therefore, by the same argument as above,
  \begin{align*}
    \sup_{\tau_1\in\RST(\mu_1)} \E_{\ol\W}[c_\alpha(\omega,\tau_1)] & = \sup_{(\tau,Y)\in \TRST(\mu_1)} \E_{\ol\W}[F(\omega,\tau,Y)]\\
    & = \inf_{\beta\geq 0}\sup_{Y\sim \mu_1} \E_{\ol\W}[F_\beta(\omega,Y)],
  \end{align*}
  where $F_\beta(\omega,y)=\sup_{0\leq t\leq L} F(\omega,t,y) + \beta(V_1-\zeta^1_{t_1})$ is upper semicontinuous and bounded from above. The last supremum is the primal problem of a classical optimal transport problem written in a probabilistic fashion. Hence, employing the classical duality result, e.g.\ \cite[Section~5]{Vi09}, we obtain
  \begin{align*}
    &  \sup_{\tau_1\in\RST(\mu_1)} \E_{\ol\W}[c_\alpha(\omega,\tau_1)]\\
    &\quad =\inf_{\beta\geq 0}\inf\left\{\int m\dd\W+\int \psi\dd\mu_1 : m\in C_b(\R_+),\psi\in C_b(\R),m(\omega)+\psi(y) \geq F_\beta(\omega,y)\right\}\\
    &\quad \geq  \inf\left\{\int m\dd\W+\int \psi\dd\mu_1 : \begin{array}{l}
                                                 \exists \beta\geq 0, m\in C_b(C(\R_+)), \psi\in C_b(\R) \text{ s.t. }\\
						  m(\omega)+\psi(y)-\beta(V_1-\zeta^1_{t}(\omega)) \geq F(\omega,t,y)
                                                \end{array}
     \right\}\\ 
    &\quad= \inf\left\{\int m\dd\W+\int \psi\dd\mu_1 : \begin{array}{l}
                                                 \exists \beta\geq 0, m\in C_b(C(\R_+)), \psi\in C_b(\R) \text{ s.t. }\\
						  m(\omega)+\psi(\omega(t))-\beta(V_1-\zeta^1_{t}(\omega)) \geq c_\alpha(\omega,t)
                                                \end{array}
    \right\}.
  \end{align*}
  Putting everything together yields the result.
\end{proof}

\begin{proof}[Proof of Theorem \ref{thm:duality}]
  By \cite[Proof of Theorem~5.10~(i), step~5]{Vi09} we can assume that $\gamma$ is continuous and bounded. We will show the result inductively by including more and more constraints (respectively \ Lagrange multipliers) in the duality result Theorem~\ref{thm:SEPDual}. In fact, we will only show the result for the two cases $n=2, I=\{2\}$ and $n=|I|=2$. The general claim follows then by an iterative application of the arguments that lead to Theorem~\ref{thm:auxDual} and the arguments below. We first consider the case where $n=|I|=2$.

  %   Observe that, by Lemma~\ref{lem:PRSTn-1}, we have
  %   \begin{align*}
  %     \sup_{(\tau_1,\tau_2)\in \RST_2(\mu_1,\mu_2)} \E_{\ol\W}[\gamma(\omega,\tau_1,\tau_2)]
  %     = \sup_{(\tau_1,\tau_2)\in\PRST_{1,2}(\mu_1,\mu_2)}\E_{\ol\W}[\gamma(\omega,\tau_1,\tau_2)].
  %   \end{align*}
  %   Moreover, by the characterising property of pseudo-random times, we can write this second term as an infimum over \emph{random times} which satisfy a martingale constraint, and a distributional constraint. Hence, we have using Theorem \ref{minmax} and the notation $X^M_t(\omega):=\E_{\ol\W}[X|\F_t^0](\omega)$ that
  
  Recall from \eqref{eq:STchar} that a random time $\tau$ is a stopping time if and only if $ \E_{\ol\W}[f(\tau)(g-\E_{\W}[g|\F_t])]=0$ for all $g\in C_b(C(\R_+)))$ and $f\in C(\R_+)$ supported on $[0,t]$. We write $H$ for the set of all functions $h\colon C(\R_+)\times \R_+\to \R$ such that $h(\omega,s)=\sum_{i=1}^n f_i(s)(g_i-\E_{\W}[g_i|\F_{u_i}])(\omega)$ for $n\in\N$, $g_i\in C_b(C(\R_+)),$ and $f_i\in C_b(\R_+)$ supported on $[0,u_i]$. Then applying Theorem~\ref{minmax} again we have

  \begin{align*}
    &\sup_{(\tau_1,\tau_2)\in \RST_2(\mu_1,\mu_2)}  \E_{\ol\W}[\gamma\circ r_2(\omega,\tau_1,\tau_2)]\\
    & \qquad = \sup_{\substack{\tau_1\in\RST(\mu_1)\\ (\tau_1,\tau_2)\in\RT_2\\ \E_{\ol\W}[\zeta_{\tau_2}^2]\leq V_2}}\ \inf_{\substack{\psi_2\in C_b(\R)\\ h\in H}} \ \E_{\ol\W}\left[\gamma\circ r_2(\omega,\tau_1,\tau_2) +h(\omega,\tau_2) - \psi_2(\omega(\tau_2)) + \int \psi_2 \dd \mu_2 \right]\\
    & \qquad =\inf_{\substack{\psi_2\in C_b(\R)\\ h\in H}}\ \sup_{\substack{\tau_1\in\RST(\mu_1)\\(\tau_1,\tau_2)\in\RT_2 \\ \E_{\ol\W}[\zeta_{\tau_2}^2]\leq V_2}} \ \E_{\ol\W}\left[\gamma_{\psi_2,h}(\omega,\tau_1,\tau_2) \right],
  \end{align*}
  where we set
  \begin{equation*}
    \gamma_{\psi_2,h}(\omega,t_1,t_2):=\gamma\circ r_2(\omega,t_1,t_2) +
    h(\omega,t_2) - \psi_2(\omega(t_2)) + \int \psi_2 \dd \mu_2 \in
    C_b\left(C(\R_+) \times \Delta_2\right).
  \end{equation*}
  Applying Theorem~\ref{thm:auxDual}, we get
  \begin{align*}
    & \sup_{(\tau_1,\tau_2)\in \RST_2(\mu_1,\mu_2)} \E_{\ol\W}[\gamma\circ r_2(\omega,\tau_1,\tau_2)]\\
    &=\inf_{\substack{\psi_2\in C_b(\R) \\ h\in H}} \ \inf\left\{ \int \psi_1 \dd \mu_1: 
       \begin{array}{l}
         \psi_1\in C_b(\R) \text{ s.t. } \\ \exists m\in C_b(C(\R_+)), \E_\W[m]=0,   \alpha_1,\alpha_2\geq 0
         \text{ s. t. }\\ m(\omega) + \psi_1(\omega(t_1)) - \sum_{i=1}^2\alpha_i(V_i-\zeta^i_{t_i}(\omega)) \geq \gamma_{\psi_2,h}(\omega,t_1,t_2)
       \end{array}\right\}.
  \end{align*}
  Take $m,\psi_1,\alpha_1,\alpha_2$ satisfying 
  \begin{equation}\label{eq:DC1}
    m(\omega) + \psi_1(\omega(t_1)) - \sum_{i=1}^2\alpha_i(V_i-\zeta^i_{t_i}(\omega)) \geq \gamma_{\psi_2,h}(\omega,t_1,t_2).
  \end{equation}
  Observe that $\E_\W[f(t)(g-\E_\W[g|\F_u])|\F_t]=0$ whenever $\supp(f)\subset [0,u].$ Fixing $t_1$ and $t_2$ inequality~\eqref{eq:DC1} can be seen as an inequality between functions of $\omega$. Hence, taking conditional expectations with respect to $\F_{t_2}$ in the sense of Definition~\ref{EAverage} and using the optionality of $\gamma$ yields
  \begin{equation*}
    \E_\W[m|\F_{t_2}](\omega) + \sum_{i=1}^2\psi_i(\omega(t_i)) - \int\psi_2\dd\mu_2 - \sum_{i=1}^2\alpha_i(V_i-\zeta^i_{t_i}(\omega)) \geq \gamma\circ r_2(\omega,t_1,t_2).
  \end{equation*}
  Hence,
  \begin{align*}
    & \sup_{(\tau_1,\tau_2)\in \RST_2(\mu_1,\mu_2)} \E_{\ol\W}[\gamma\circ r_2(\omega,\tau_1,\tau_2)]\\
    &\geq \inf_{\psi_2\in C_b(\R)} \ \inf\left\{ \int \psi_1 \dd \mu_1+\int\psi_2 \dd \mu_2: 
       \begin{array}{l}
         \text{there exist a $\S$-continuous martingale }M,\\
         M_0=0, \psi_1\in C_b(\R_+)\text{ and } \alpha_1,\alpha_2\geq 0 \text{ s.t. }\\
         \sum_{i=1}^2(\psi_i(\omega(t_i))+M_{t_2}(\omega))\\ - \sum_{i=1}^2\alpha_i(V_i-\phi_i(\omega(t_i)) + \phi_i(\omega(t_i))-\zeta_{t_i}^i(\omega)) \\ \quad \quad \geq \gamma\circ r_2(\omega,t_1,t_2)
       \end{array}\right\}\\
    &= \inf_{\psi_1,\psi_2\in \mathcal{E}_1\times\mathcal E_2}\left\{\int \psi_1 \dd \mu_1+\int\psi_2 \dd \mu_2: 
       \begin{array}{l}
         \text{there exist two $\S$-continuous martingales }M^i,\\ M^i_0 =0,
         \text{ s.t. } \\\sum_{i=1}^2(\psi_i(\omega(t_i))+M^i_{t_i}(\omega))  \geq \gamma\circ r_2(\omega,t_1,t_2)
       \end{array}\right\} \\ 
    &= \ D^*_2,
  \end{align*}
  where in the final step we used the fact that $\E_{\ol\W}[\phi_i(B_{\tau_i})]=\E_{\ol\W}[\zeta_{\tau_i}^i]$, $\int \phi_i \, \dd \mu_i =V_i$, $\phi_i(B_0)=0$, and that $\phi_i(B)-\zeta^i$ is a martingale. 

  For later use, we write:
  \begin{equation*}
    D(\gamma):= \left\{(\psi_1,\psi_2) \in \cE_1\times \cE_2: 
       \begin{array}{l}
         \text{there exist two $\S$-continuous martingales } M^i,\\
         M^i_0 =0, \text{ s.t.}\\\sum_{i=1}^2(\psi_i(\omega(t_i))+M^i_{t_i}(\omega))  \geq \gamma\circ r_2(\omega,t_1,t_2)
       \end{array}\right\}.
  \end{equation*}

  We now consider the case where $n=2, |I|=1$ and $I = \{2\}$, so we are prescribing $\mu_2$ but not $\mu_1$. Writing $\rho \preceq \nu$ to denote that $\rho$ precedes $\nu$ in convex order, we use the result of the case where $|I|=2$ to see that:
  \begin{align*}
    P^*_2 = \sup_{(\tau_1,\tau_2)\in \RST_2(\mu_2)}  \E_{\ol\W}[\gamma\circ r_2(\omega,\tau_1,\tau_2)]
    & = \sup_{\mu_1 \preceq \mu_2} \sup_{(\tau_1,\tau_2)\in \RST_2(\mu_1,\mu_2)}  \E_{\ol\W}[\gamma\circ r_2(\omega,\tau_1,\tau_2)]\\
    & = \sup_{\mu_1 \preceq  \mu_2} \inf_{(\psi_1,\psi_2) \in D(\gamma)}\left\{\int \psi_1
      \dd\mu_1+\int\psi_2 \dd\mu_2\right\}.
  \end{align*}
  We now need to introduce some additional compactness. Recall from the definitions of $\phi_i$ that $\phi_2/\phi_1 \to \infty$ as $x \to \pm \infty$. Now let $\eps>0$ and write
  \begin{equation*}
    D^\eps(\gamma^\eps):= \left\{(\psi_1^\eps,\psi_2) : 
       \begin{array}{l}
         \psi_1^\eps + \eps \phi_2 \in \cE_1, \psi_2 \in \mathcal{E}_2, \text{ and there exist two
         $\S$-continuous} \\ \text{martingales } M^i, M^i_0=0
         \text{ such that:}\\\psi_1^\eps(\omega(t_1)) + \psi_2(\omega(t_2)) +\sum_{i=1}^2
         M^i_{t_i}(\omega))  \geq \gamma^\eps\circ r_2(\omega,t_1,t_2)
       \end{array}\right\}.
  \end{equation*}
  In particular, we have $(\psi_1,\psi_2) \in D(\gamma) \iff (\psi_1-\eps \phi_2,\psi_2) \in D^\eps(\gamma-\eps \phi_2(\omega(t_1)))$ and so (with $\psi_1^\eps = \psi_1-\eps \phi_2, \gamma^\eps = \gamma-\eps\phi_2(\omega(t_1))$)
  \begin{align*}
    \inf_{(\psi_1,\psi_2) \in D(\gamma)}\left\{\int \psi_1
    \dd\mu_1+\int\psi_2 \dd\mu_2\right\} & = \inf_{(\psi_1^\eps,\psi_2) \in
       D^\eps(\gamma^\eps)}\left\{\int (\psi_1^\eps+\eps \phi_2)
      \dd\mu_1+\int\psi_2 \dd\mu_2\right\}\\
     =& \inf_{(\psi_1^\eps,\psi_2) \in
      D^\eps(\gamma^\eps)}\left\{\int \psi_1^\eps
       \dd\mu_1+\int\psi_2 \dd\mu_2\right\} + \eps  \int \phi_2\, \mu_1(\d x).
  \end{align*}
  In particular, the final integral can be bounded over the set of $\mu_1 \preceq \mu_2$, and so by taking $\eps>0$ small, this term can be made arbitrarily small. Moreover, by neglecting it we get a quantity that is smaller than $P$.

  If we introduce the set 
  \begin{equation*}
    \textsf{CV}: = \big\{c:\R\to\R: c \text{ convex, }  c(x) \ge 0, c \text{ smooth}, c(x) \le L(1+|x|) , \text{ some } L \ge 0\big\},
  \end{equation*}
  then we may test the convex ordering property by penalising against $\textsf{CV}$. In particular, we can write after another application of Theorem \ref{minmax}
  \begin{align*}
    P^*_2\ge \inf_{(\psi_1^\eps,\psi_2) \in D^\eps(\gamma^\eps)} & \ \sup_{\mu_1 \preceq  \mu_2} \left\{\int \psi_1^\eps
        \dd\mu_1+\int\psi_2 \dd\mu_2\right\}\\
      & = \inf_{(\psi_1^\eps,\psi_2) \in D^\eps(\gamma^\eps)} \sup_{\mu_1} \inf_{\substack{c \in \textsf{CV}}} \left\{\int (\psi_1^\eps-c)
    \dd\mu_1+\int(\psi_2+c) \dd\mu_2\right\}.
  \end{align*}
  In addition, for fixed $\psi^\eps_1 \in D^\eps(\gamma^\eps)$, we observe that, by the fact that $\psi^\eps_1+\eps\phi_2 \in \mathcal{E}_1$, we must have $\psi^\eps_1(x) \to -\infty$ as $x \to \pm \infty$. Hence, we can find a constant $K$, which may depend on $\psi_1^\eps$, so that $\psi^\eps_1(x) <\psi_1^\eps(0)$ for all $x \not\in [-K,K]$. In particular, we may restrict the supremum over measures $\mu_1$ above to the set of probability measures $\mathcal{P}_K:= \{\mu: \mu([-K,K]^c) = 0\}$, where $A^c$ denotes the complement of the set $A$. Note that this set is compact, so we can then apply Theorem~\ref{minmax} to get:
  \begin{align*}
    \inf_{(\psi_1^\eps,\psi_2) \in D^\eps(\gamma^\eps)} 
    & \ \sup_{\mu_1 \preceq  \mu_2} \left\{\int \psi_1^\eps \dd\mu_1+\int\psi_2 \dd\mu_2\right\}\\
    & = \inf_{(\psi_1^\eps,\psi_2) \in D^\eps(\gamma^\eps)} \inf_{\substack{c \in
      \textsf{CV}}} \ \sup_{\mu_1\in \mathcal{P}_K}\  \left\{\int (\psi_1^\eps-c)
    \dd\mu_1+\int(\psi_2+c) \dd\mu_2\right\}\\
    & = \inf_{(\psi_1^\eps,\psi_2) \in D^\eps(\gamma^\eps)} \ \inf_{\substack{c \in
      \textsf{CV}}} \ \left\{ \sup_{x \in [-K,K]} \left[\psi_1^\eps(x)-c(x)\right]
    +\int(\psi_2+c) \dd\mu_2\right\}.
  \end{align*}
  In particular, for any $\delta>0$, we can find $(\psi_1^\eps,\psi_2) \in D^\eps(\gamma^\eps)$ and $c \in \textsf{CV}$ such that
  \begin{equation*}
    P^*_2 \ge \sup_{x \in \R} \left[\psi_1^\eps(x)-c(x)\right]
    +\int (\psi_2+c) \dd \mu_2 - \delta.
  \end{equation*}
  Take $\psi_2^\eps(\omega(t_2)) := \sup_{x \in \R} \left[\psi_1^\eps(x)-c(x)\right] + \psi_2(\omega(t_2))+c(\omega(t_2)) +\eps\phi_2(\omega(t_2))$. Then there exist $M^1, M^2$ such that
  \begin{align*}
    \gamma^\eps\circ r_2(\omega,t_1,t_2) & \ \le\  \psi_1^\eps(\omega(t_1)) + \psi_2(\omega(t_2)) + \sum_{i=1}^2 M^i_{t_i}(\omega)\\
     & \ = \ \psi_2^\eps(\omega(t_2)) + \sum_{i=1}^2
                                                  M^i_{t_i}(\omega)-\eps \phi_2(\omega(t_2))-c(\omega(t_2)) +
                                                  c(\omega(t_1))\\  
                                & \qquad {}+ \left[\psi_1^\eps(\omega(t_1))-
                                                  c(\omega(t_1))\right] - \sup_{x \in \R} \left[\psi_1^\eps(x)-c(x)\right].
  \end{align*}
  Hence,
  \begin{align*}
    \gamma\circ r_2(\omega,t_1,t_2)
     & \ \le \ \psi_2^\eps(\omega(t_2)) + \sum_{i=1}^2
                                                M^i_{t_i}(\omega)+ \eps (\phi_2(\omega(t_1)) -
                                                \phi_2(\omega(t_2))) \\ 
                                & \qquad {} -c(\omega(t_2)) +  c(\omega(t_1))\\ 
     & \ =\  \psi_2^\eps(\omega(t_2)) + \sum_{i=1}^2
                                                M^i_{t_i}(\omega)\\ 
                                & \qquad {} + \eps\left[ (\phi_2(\omega(t_1))-\zeta_{t_1}^2) -
                                  (\phi_2(\omega(t_2))-\zeta_{t_2}^2)\right] + 
                                  \eps(\zeta_{t_1}^2 - \zeta_{t_2}^2)\\
                                & \qquad {} + \left[ (c(\omega(t_1))-\zeta_{t_1}^c) -
                                  (c(\omega(t_2))-\zeta_{t_2}^c)\right] + 
                                  (\zeta_{t_1}^c - \zeta_{t_2}^c).
  \end{align*}
  Since $\zeta_t^2$ is an increasing process, compensating $\phi_2$, then $\zeta_{t_2}-\zeta_{t_1} \ge 0$ whenever $t_1 \le t_2$. Similarly, $\zeta^c_t$ is the increasing process compensating $c$, and the same argument as above holds. Note that $\zeta^c$ is $\S$-continuous since $c$ is assumed smooth. It follows that $(\psi_1^\eps,\psi_2) \in D^\eps(\gamma^\eps)$ implies $\psi_2^\eps \in D'(\gamma)$, where
  \begin{equation*}
    D'(\gamma):= \left\{\psi_2 \in \mathcal{E}_2: 
       \begin{array}{l}
         \text{there exist two $\S$-continuous martingales $M^i, M^i_0 =0$}\\
         \text{such that } \psi_2(\omega(t_2)) +\sum_{i=1}^2 M^i_{t_i}(\omega)  \geq \gamma\circ r_2(\omega,t_1,t_2)
       \end{array}\right\}.
  \end{equation*}
  It follows by making $\eps, \delta$ small that
  \begin{align*}
    P^*_2 \ge \inf_{\psi_2 \in D'(\gamma)} \int \psi_2 \dd\mu_2(x),
  \end{align*}
  and as usual, the inequality in the other direction is easy.\medskip 
  
  To establish the claim in the general case we can now successively introduce more and more constraints accounting for more and more Lagrange multipliers and use either only the first or the first and the second argument to prove the full claim.
\end{proof}

To conclude, we can follow the reasoning of Section~\ref{sec:result}, more precisely Step~1 and Step~3, and obtain the following robust super-hedging result: 
\begin{theorem}\label{thm:generalduality}
  Suppose that $n\in\mathbb{N}$, $I\subseteq \{1,\ldots, n\}$, $n\in I$ and that $\mu_i$ is a centered probability measure on $\mathbb{R}$ for each $i\in I$ and let $G\colon C[0,n]\to \R$ be of the form
\begin{align}\label{invform}
  G(\omega)=\gamma( \ntt(\omega)_{\llcorner [0, \langle \omega\rangle_{n}] }, \langle \omega\rangle_{1}, \ldots, \langle \omega\rangle_{n}),
\end{align}
where $\gamma $ is $\S_n$-upper semi-continuous and bounded from above.  Let us define 
  \begin{align*}
    P_n:=\sup\big\{ {\E}_\P[G]: \P \text{ is a martingale measure on } C[0,n],\, S_0=0,\,  S_{i}\sim \mu_i\text{ for all }i\in I\big\}
  \end{align*}
  and 
  \begin{align*}
    D_n:=\inf\left\{a: \begin{array}{l} \exists c>0,  f \in C^\infty(\R,\R) \text{ s.t. } |f| /(1+\varphi_n) \text{ is bounded}, \\
    			(H^m)_{m \in \N} \subseteq \mathcal{Q}_{f,c} \text{ and } (\psi_j)_{j\in I}, \int \psi_j \dd\mu_j=0 \text{ s.t. }\forall \omega \in C_{\mathrm{qv}}[0,n]\\
                         a+\sum_{j\in I} \psi_j(S_j(\omega)) + \liminf_{m\to \infty}(H^m\cdot S)_n(\omega) \geq G(\omega)
                      \end{array}
            \right\},
  \end{align*}
   where for $f \in C^2(\R,\R)$ we set
  \[
  \mathcal{Q}_{f,c} :=\big\{H: H \text{ is a simple strategy and } (H \cdot S)_t(\omega) \ge -c - \zeta^f_t(\omega) \ \forall (\omega,t) \in C_{\mathrm{qv}}[0,n] \times [0,n] \big\}.
  \]
  
  Under the above assumptions we have $P_n=D_n$.
\end{theorem}

Finally, we note that Theorem \ref{thm:generalduality} could be further extended based on the above arguments. For example, we could include additional market information on prices of further options of the invariant form \eqref{invform}.

\bibliography{joint_biblio}{}
\bibliographystyle{plain}

\end{document}